\documentclass[conference]{IEEEtran}
\usepackage[
	paper=a4paper,
	margin=51pt,
	includeheadfoot
]{geometry}
\usepackage{wrapfig}
\usepackage{amsfonts}
\usepackage{amssymb,epic,eepic}
\usepackage{amsmath}
\usepackage{dsfont}
\usepackage{dcolumn}
\usepackage{bm}
\usepackage{bbm}
\usepackage{verbatim}
\usepackage{color}
\usepackage{amsthm}
\usepackage{tikz}
\usepackage{epstopdf}
\usetikzlibrary{patterns,positioning,decorations.pathreplacing,angles,quotes}
\usepackage{MnSymbol,wasysym}
\usetikzlibrary{arrows,positioning,decorations.pathmorphing,decorations.markings,shapes,fadings}
\usepackage{pstool}
\usepackage{psfrag}
\usepackage{import}
\usepackage{varwidth}
\usepackage{nicefrac}
\usepackage{subcaption}
\usepackage[font=footnotesize]{caption}
\newcommand{\cX}{{\mathcal X}}

\newcommand{\nn}{{\mathbb N}}

\newcommand{\eps}{{\varepsilon}}        

\newcommand{\cP}{\mathcal P}

\newcommand{\cC}{\mathcal C}

\newtheorem{theorem}{Theorem}

\newtheorem*{acknowledgement}{Acknowledgement}

\newtheorem{definition}{Definition}

\newtheorem{lemma}{Lemma}

\newtheorem{remark}{Remark}

\newcommand{\tr}{\mathrm{tr}}

\makeatletter
\newcommand*{\rom}[1]{\expandafter\@slowromancap\romannumeral #1@}
\makeatother
\providecommand{\keywords}[1]{{\fontsize{9}{0}\textbf{\textit{Index Terms---}#1}}}
\DeclareMathOperator{\conv}{conv}

\DeclareMathOperator{\linspan}{span}

\DeclareRobustCommand*{\IEEEauthorrefmark}[1]{\raisebox{0pt}[0pt][0pt]{\textsuperscript{\footnotesize #1}}}
\tikzset{
	schraffiert/.style={pattern=horizontal lines,pattern color=#1},
	schraffiert/.default=black
}
\usetikzlibrary{patterns}

\newcommand\copyrighttext{%
	\footnotesize\copyright 2019 IEEE. Personal use of this material is permitted. Permission from IEEE must be obtained for all other uses, in any current or future media, includingreprinting/republishing this material for advertising orpromotional purposes, creating new collective works, for resale or redistribution to servers or lists, orreuse of any copyrighted component of this work in other works.}
\newcommand\copyrightnotice{%
	\begin{tikzpicture}[remember picture,overlay]
	\node[anchor=south,yshift=-20pt] at (current page.south) {\fbox{\parbox{\dimexpr\textwidth-\fboxsep-\fboxrule\relax}{\copyrighttext}}};
	\end{tikzpicture}%
}

\begin{document}
	\title{Reliable Communication over Arbitrarily Varying Channels under Block-Restricted Jamming}
	\author{Christian Arendt\IEEEauthorrefmark{1,2}, Janis N\"otzel\IEEEauthorrefmark{3} and Holger Boche\IEEEauthorrefmark{2}\\
		\scriptsize
		\IEEEauthorblockA{\IEEEauthorrefmark{1}
			BMW Group, 80788 M\"unchen, Germany, Email: christian.ca.arendt@bmw.de}
		
		\IEEEauthorblockA{\IEEEauthorrefmark{2}
			Lehrstuhl f\"ur Theoretische Informationstechnik, Technische Universit\"at M\"unchen, 80290 M\"unchen, Germany}
		\IEEEauthorblockA{\IEEEauthorrefmark{3}
			Theoretische Nachrichtentechnik, Technische Universit\"at Dresden, 01187 Dresden, Germany}\\
	}
	\maketitle
	\copyrightnotice
	\begin{abstract}
	We study reliable communication in uncoordinated vehicular communication from the perspective of Shannon theory. Our system model for the information transmission is that of an Arbitrarily Varying Channel (AVC): One sender-receiver pair wants to communicate reliably, no matter what the input of a second sender is. The second sender is assumed to be uncoordinated and interfering, but is supposed to follow the rational goal of transmitting information otherwise. We prove that repetition coding can increase the capacity of such a system by relating the notion of symmetrizability of an arbitrarily varying channel to invertibility of the corresponding channel matrix. Explicit upper bounds on the number of repetitions needed to prevent system breakdown through diversity are provided. Further we introduce the notion of block-restricted jamming and present a lower and an upper bound on the maximum error capacity of the corresponding restricted AVC.
	\end{abstract}
	{\smallskip}\keywords{Reliable communication, unknown interference, arbitrarily varying channel, block-restricted jamming}.
\begin{section}{Introduction}
	Reliable communication in unlicensed frequency bands is one of the key challenges in wireless connectivity. Especially when looking at reliability issues in state-of-the-art distributed random access schemes which can, for example, be found in Vehicular Ad-hoc NETworks (VANETs), the necessity of increasing the stability of a wireless link is indisputable. Today there exist two major standards for direct vehicular communication in VANETs, IEEE 802.11p (11p) \cite{5514475} and Device-to-Device (D2D) communication in Long Term Evolution (LTE) for vehicle-to-vehicle communication (designated as PC5 \cite{7859501}). In both standards there may occur situations where no central entity controls the use of spectral resources. Thus the frequency band is used in a shared, self-coordinated manner. Power control and adaption techniques cannot fully prevent simultaneous channel access. Thus there may arise situations where an unknown information source causes interference to a point-to-point link. In particular, coexistence of different technologies is a non-negligible issue. Unknown interference can cause packet collisions resulting in packet losses in actual communication schemes, for example in 11p, where collided packets are dropped \cite{7313128}. The necessity of integrating reliability requirements in the physical domain of wireless communication motivates the information-theoretic investigations in this contribution.

	In information theory, communication over a channel with unknown interference is modeled by an Arbitrarily Varying Channel (AVC), introduced by Blackwell et al. in \cite{blackwell1960}. In the AVC model, interference is incorporated by introducing a jammer controlling the channel state in an arbitrary manner. A special feature of the AVC under average error criterion is the effect of symmetrizability, that is, the jammer may choose his state inputs such that any two codewords may be confused by the decoder. For symmetrizable AVCs reliable communication at positive rates cannot be guaranteed. Csiszár and Narayan deduced in \cite{2627} that non-symmetrizability is a sufficient condition for communication over an AVC at positive rates under the average error criterion using deterministic coding.
	
	Motivated by high reliability specifications in the framework of reliable communication in future communication standards, it is important to enable information exchange under maximum error requirements. The capacities of a Discrete Memoryless Channel (DMC) under maximum and average error criterion are identical. In the AVC setting, sender and receiver are lacking information regarding the channel state. Thus the encoder and decoder, as well as the codebook, have to be chosen such that they allow reliable message transmission over a large and potentially varying number of possible channel laws. In this context, the capacities for message transmission under maximum and average error are not necessarily identical \cite{720535}. For the maximum error capacity, a symmetrizability statement is given by Ahlswede in \cite{Ahlswede1978}. The same author showed in \cite{AHLSWEDE1969457} that communication over AVCs at positive rates can be possible using Common Randomness (CR)-assisted coding, even when it is impossible without. CR-assisted coding can be established by two legitimate communication parties observing correlated outcomes of a random experiment hidden from the jammer. However, CR-assisted communication so far requires side links or a common reference signal. In this work, we present diversity as an alternative enabler for reliable communication over AVCs. In contrast to \cite{arendt2017}, where we concentrate on spatial receive diversity, we here focus on transmit diversity in time domain. Additionally, we apply the maximum error leading to a more stringent performance analysis compared to the average error analysis conducted in \cite{arendt2017}.
\linebreak\\
	{\bf Outline of the paper.}  We introduce the notation, coding concepts, as well as the AVC model together with the symmetrizability conditions in Section~\ref{sec:notation_and_channel_models}. In Section~\ref{sec:main_Results} we demonstrate that injectivity of the channel matrix implies non-symmetrizability under both error criteria. This phenomenon is subsequently investigated in the following. Furthermore, we relate the result to the positivity of the maximum error capacity for deterministic coding and provide a computable lower bound on the maximum error capacity if it is positive. In Section~\ref{sec:temporal_prop} we introduce the AVC under block-restricted jamming and present a lower and an upper bound on its maximum error capacity. The proofs of supplementary results (lemmas) are postponed to the appendix.
\end{section}
\begin{section}{Notation, Definitions and Channel Models}\label{sec:notation_and_channel_models}
	We adapt our notation to the one presented in \cite{7447794,arendt2017}: For $L\in\nn$, we define $[L]:=\{1,\ldots,L\}$. We denote the set of permutations on $[L]$ by $S_L$. Let two sets $\mathcal{X},\mathcal{Y}$ of cardinality $|\mathcal X|=L$ and $|\mathcal Y|=L'$ with $L,L'\in\nn$ be given. Their product is given by $\mathcal{X}\times\mathcal{Y}:=\{(x,y):x\in\mathcal{X},\ y\in\mathcal{Y}\}$. Additionally, $\mathcal{X}^n$ is the n-fold product of $\mathcal{X}$ with itself for any $n\in\nn$. The set of probability distributions on a finite set $\mathcal{X}$ is denoted by
	\begin{align}
		\cP(\mathcal{X}):=\left\{p:\mathcal{X}\to\mathbb R\ :\ p(x)\geq 0\ \forall\ x\in\mathcal{X},\ \sum_{x\in\mathcal{X}}p(x)=1\right\}.
	\end{align}
	Further, for all $J\in\mathbb{N}$ we set
	\begin{align}
		\cP^{(J)}(\mathcal{X}):=\left\{\sum_{i=1}^Jp(i)\delta_i^{\otimes J}:p\in\cP(\mathcal{X})\right\}.
	\end{align}
	An important subset of elements of $\cP(\mathcal{X})$ is the set of its extremal points, the Dirac-measures: For $x,x'\in\mathcal{X}$, $\delta_x\in\cP(\mathcal{X})$ is defined through $\delta_x(x')=\delta(x,x')$ where $\delta(\cdot,\cdot)$ is the usual Kronecker-delta symbol. We transfer the probabilistic concepts to linear algebra by considering $\cP(\mathcal{X})$ as being embedded into $\mathbb R^L$ through the bijection $p\mapsto \sum_{x\in\mathcal{X}}p(x)e_x$.	Under this transformation, $\delta_x$ is mapped to $e_x$. This allows a natural use of matrix calculus in our analysis. We solely introduce results from multi-linear algebra for bipartite systems. The generalization to the multi-partite case is straightforward. We use fixed bases $\{e_x\}_{x\in\mathcal X}$, $\{e_y\}_{y\in\mathcal Y}$ for $\mathbb R^L$ and $\mathbb R^{L'}$. $L\times L'$ matrices define linear maps from $\mathbb R^L$ to $\mathbb R^{L'}$ via their actions in these bases. The scalar product $\langle\cdot,\cdot\rangle$ on $\mathbb R^L\times\mathbb R^{L}$ is the standard one: $\langle e_x,e_{x'}\rangle=\delta(x,x')$.
	
	The tensor product of $\mathbb R^L$ with $\mathbb R^{L'}$ is
	\begin{align}
		\mathbb R^L\otimes\mathbb R^{L'}:=\linspan\{e_x\otimes e_{y}\}_{x,\in\mathcal{X},y\in\mathcal{Y}}.
	\end{align}
	This allows us to define general ``product vectors" of two vectors $u=\sum_{x\in\mathcal{X}}u_xe_x$ and $v=\sum_{y\in\mathcal{Y}}v_ye_y$ by
	\begin{align}
		u\otimes v:=\sum_{x\in\mathcal{X},y\in\mathcal{Y}}u_xv_{y}e_x\otimes e_{y}.
	\end{align}
	The vector space $\mathbb R^L\otimes\mathbb R^{L'}$ inherits the scalar product by the formula $\langle u\otimes v,x\otimes y\rangle:=\langle u,x\rangle\langle v,y\rangle$. The space of $L\times L'$ matrices is denoted by $M_{L\times L'}$. Given $A,B\in M_{L\times L'}$, we define $A\otimes B$ through its action on product vectors:
	\begin{align}
		(A\otimes B)(u\otimes v):=(Au)\otimes(Bv).
	\end{align}
	In order to simplify notation later, for $u\in\mathbb R^L$ and $n\in\nn$ we will use the shorthand $u^{\otimes n}:=u\otimes\ldots\otimes u$ for the $n$- fold tensor product of $u$ with itself. Accordingly, for $A\in M_{L\times L'}$, we write $A^{\otimes n}:=A\otimes\ldots\otimes A$. The partial trace $\tr_{[L']}:\mathbb R^L\otimes\mathbb R^{L'}\to\mathbb R^L$ summing over the ``content" of $\mathbb R^{L'}$ is defined in the following way: For $v=\sum_{i,j=1}^{L,L'}v_{i,j}\delta_i\otimes \delta_j$, the partial trace operator is defined as $\tr_{[L']}(v):=\sum_{i,j=1}^{L,L'}v_{i,j}\delta_j$.
	
	The influence of noise during transmission of messages is modeled by stochastic matrices $W$ of conditional probability distributions $\left(w\left(y\vert x\right)\right)_{x\in\mathcal{X},y\in\mathcal{Y}}$, whose entries satisfy $\forall x\in\mathcal{X}$, $y\in\mathcal{Y}:\ w(y\vert x)\in\left[0,1\right]$ and $\forall x\in\mathcal{X}:\ w\left(\cdot\vert x\right)\in\cP(\mathcal{Y})$. Any such matrix is henceforth also called a channel. The set of channels acting on a finite alphabet $\mathcal{X}$ of size $L$ and $\mathcal{Y}$ of size $L'$ is denoted by $\cC(\mathcal{X},\mathcal{Y})$. The special case where $\forall x\in\mathcal{X}$, $y\in\mathcal{Y}:\ w(y\vert x)=\delta(y,x)$ is denoted by $Id$.

	In later analysis, we make use of the Shannon entropy of $p\in\mathcal{P}(\mathcal{X})$ which is defined as
	\begin{align}
		H(p):=-\sum_{x\in\mathcal{X}}^{}p(x)\log(p(x)).
	\end{align}
	Every channel $W:\mathcal{P}(\mathcal{X})\mapsto\mathcal{P}(\mathcal{Y})$ together with a probability distribution $p\in\mathcal{P}(X)$ defines a joint distribution $\mathcal{P}((X,Y)=(x,y))=p(x)w(y\vert x)$ for all $x\in\mathcal{X}$ and $y\in\mathcal{Y}$. The mutual information which, by default, is defined as $I(X;Y):=H(X)-H(X\vert Y)$, can then equivalently be written as $I(p;W):=I(X,Y)$.
	
	In order to understand the cause of system breakdowns due to denial of service attacks or unknown interference, it is important to accurately model these effects in a probabilistic framework. For this reason, we focus on AVCs in particular. The probabilistic law governing the transmission of codewords over a point-to-point AVC for $n$ channel uses is described by
	\begin{align}
		w^{\otimes n}\left(y^n\vert x^n,s^n\right):=\prod\limits_{i=1}^nw(y_i\vert x_i,s_i),
	\end{align}
	where $s^n=(s_1,\ldots,s_n)\in\mathcal{S}^n$ are the jammer's state inputs, $x^n=(x_1,\ldots,x_n)\in\mathcal{X}^n$ are the input codewords of the encoder and $y^n=(y_1,\ldots,y_n)\in\mathcal{Y}^n$ denote the channel outputs at the decoder, all assumed to be taken from finite alphabets. The previously introduced notion naturally extends to products of AVCs. Let, for example, $K=2$ DMCs $W_1\in\mathcal C(\mathcal X_1,\mathcal Y_1)$ and $W_2\in\mathcal C(\mathcal X_2,\mathcal Y_2)$ be given. Then the transition probability matrix of $W_1\otimes W_2$ is defined by $w(y_1,y_2\vert x_1,x_2):=w_1(y_1\vert x_1)\cdot w_2(y_2\vert x_2)$, for all $x_1\in\mathcal{X}_1, x_2\in\mathcal{X}_2, y_1\in\mathcal{Y}_1, y_2\in\mathcal{Y}_2$. This notation can be adapted to AVCs:
\definecolor{light-gray}{gray}{0.8}
\definecolor{dark-gray}{gray}{0.4}
\definecolor{extremelight-gray}{gray}{0.9}
\definecolor{darkblue}{RGB}{0,0,90}
\definecolor{lightblue}{RGB}{102,153,255}
\begin{figure}
	\centering
	\begin{tikzpicture}[thick,scale=0.745, every node/.style={transform shape}]
	\node (T1)  [fill=light-gray,rectangle,inner sep=7pt, text width=1.5cm, align=center] at (0.4,3) {{Encoder}};
	\node (T2)  [fill=light-gray,rectangle,inner sep=7pt, text width=1.5cm, align=center] at (7.5,3) {{Decoder}};
	\node (T3)  [fill=darkblue,rectangle,inner sep=7pt, text width=2.5cm, align=center] at (3.9,3) {\textcolor{white}{$w^{\otimes n}\left(y^n\vert x^n,s^n\right)$}};
	\node (T2)
	[fill=lightblue,rectangle,inner sep=7pt, text width=1.5cm, align=center] at (3.9,1) {{Jammer}};
	\path[-triangle 60, color=dark-gray,line width=1pt, shorten <= -1pt] (3.9,1.4) edge (3.9,2.55);
	\path[-triangle 60, color=dark-gray,line width=1pt, shorten <= -1pt] (1.45,3) edge (2.4,3);
	\path[-triangle 60, color=dark-gray,line width=1pt, shorten <= -1pt] (5.45,3) edge (6.5,3);
	\path[-triangle 60, color=dark-gray,line width=1pt, shorten <= -1pt] (-1.6,3) edge (-0.6,3);
	\path[-triangle 60, color=dark-gray,line width=1pt, shorten <= -1pt] (8.55,3) edge (9.5,3);
	\node (C1) at (3.6,1.9)[text width=2cm,align=center] {\textcolor{black}{{$S^n$}}};
	\node (C2) at (1.8,3.3)[text width=2cm,align=center] {\textcolor{black}{{$X^n$}}};
	\node (C3) at (5.85,3.3)[text width=2cm,align=center] {\textcolor{black}{{$Y_{s^n}^n$}}};
	\node (C4) at (-1.2,3.3)[align=center] {\textcolor{black}{{$M$}}};
	\node (C4) at (8.8,3.3)[text width=2cm,align=center] {\textcolor{black}{{$\hat M$}}};
	\end{tikzpicture}
	\caption{Block diagram of an AVC. The state of the channel is controlled by the state sequence $s^n\in\mathcal S^n$ chosen by the jammer.}
	\label{fig:system_model}
\end{figure}
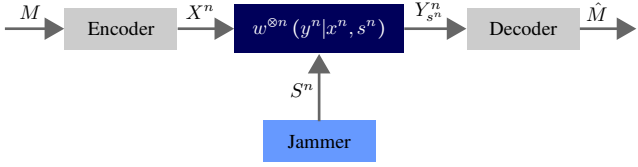
	\begin{definition}[Arbitrarily Varying Channel (AVC)]\label{def:AVC_single}
		Let $\mathcal{X}$, $\mathcal{Y}$, $\mathcal{S}$ be finite sets. Let $W_s\in\mathcal C(\mathcal X,\mathcal Y)$ for every $s\in\mathcal S$. The corresponding arbitrarily varying channel (AVC) is denoted by $\mathcal W:=(W_s)_{s\in\mathcal S}$ or, alternatively, $w(\cdot|x,s):=W_s(\delta_x)$, or $\mathcal W=(w(y|x,s))_{y\in\mathcal Y,x\in\mathcal X,s\in\mathcal S}$. The AVC is denoted by $\mathcal{W}$ and its action is completely described via the sequence $((W_{s^n})_{s^n\in\mathcal S^n})_{n\in\mathbb N}$,	where $W_{s^n}=W_{s_1}\otimes\ldots\otimes W_{s_n}$.
	\end{definition}
A block diagram of the AVC is shown in Figure~\ref{fig:system_model}.
	\begin{definition}[Unassisted Deterministic Code]\label{def:det_code}
		An unassisted (deterministic) code $\mathcal{K}^\mathrm{UA}_n$ for the AVC $\mathcal{W}\in\mathcal{C}(\mathcal{X}\times\mathcal{S},\mathcal{Y})$ consists of a set $\mathcal{M}_n$ of messages and a deterministic encoder $f:\mathcal{M}_n\to\mathcal X^n$ in combination with a collection $\{\mathcal{D}_m\}_{m=1}^{|\mathcal{M}_n|}$ of decoding subsets $\mathcal{D}_m\subset\mathcal{Y}$ for which $\mathcal{D}_m\cap \mathcal{D}_{m'}=\emptyset$ for every $m\neq m'$. The average error of the code $\mathcal{K}^\mathrm{UA}_n$ is given by
		\begin{align}
		\overline{e}_\mathrm{UA}\left(\mathcal{K}^\mathrm{UA}_n\right):=\max_{s^n\in\mathcal{S}^n}\frac{1}{|\mathcal M_n|}\sum\limits_{m=1}^{|\mathcal M_n|}w^{\otimes n}\left(\mathcal{D}^{\complement}_m\vert f(m),s^n\right).
		\end{align}
		Likewise, the maximum error of the code $\mathcal{K}_n$ is given by
		\begin{align}
		{e}^{\mathrm{max}}_{\mathrm{UA}}\left(\mathcal{K}^\mathrm{UA}_n\right):=\max_{s^n\in\mathcal{S}^n}\max_{m\in\mathcal{M}_n}w^{\otimes n}\left(\mathcal{D}^{\complement}_m\vert f(m),s^n\right).
		\end{align}
	\end{definition}
	Advanced encoding and decoding schemes rely on the access to a coordination resource, that is, a random variable $\Gamma$ shared by the transmitter and receiver. $\Gamma$ is used to coordinate the choice of encoders and decoders. For CR-assisted communication, the capacity under maximum error equals the average error capacity (cf.~\cite{csiszar2011information}).
	\begin{definition}[Achievable Rate]\label{def:achievable_rate}
		A non-negative number $R$ is called an \textit{achievable rate} for the AVC $\mathcal{W}\in\mathcal{C}(\mathcal{X}\times\mathcal{S},\mathcal{Y})$ under the average error criterion, if for every $\epsilon>0$ and $\delta>0$ and $n$ sufficiently large, there exists an unassisted code $\mathcal{K}_n$ such that $\frac{\log\vert \mathcal{M}_n\vert}{n}>R-\delta$, and $\overline{e}_\mathrm{UA}<\epsilon$. The achievable rate for $e^\mathrm{max}_\mathrm{UA}$ is defined accordingly.
	\end{definition}
	\begin{definition}\label{def:capacity}
		Let $K\in\mathbb{N}$. Given an AVC $\mathcal{W}\in\mathcal{C}(\mathcal{X}\times\mathcal{S},\mathcal{Y})$, the deterministic capacities of the AVC are defined as
		\begin{align}
		\overline{C}_\mathrm d(\mathcal W):=\sup\left\{R:\begin{array}{l}R\mathrm{\ is\ an\ achievable\ rate\ for\ a}\\
		\mathrm{deterministic\ coding\ scheme}\\
		\mathrm{under\ average\ error}
		\end{array}\right\}.
		\end{align}
	 $C_\mathrm d^\mathrm{max}$ is defined accordingly.
	\end{definition}
	\begin{definition}[Maximum Error Symmetrizability \cite{Ahlswede1978}]\label{def:x-symm}
		Let $\mathcal{W}\in\mathcal{C}(\mathcal{X}\times\mathcal{S},\mathcal{Y})$ be an AVC. $\mathcal{W}$ is maximum error symmetrizable $($${e}^{\mathrm{max}}$-symm.$)$, if for all $x,x'\in\mathcal{X}$ it holds
		\begin{align}
		\conv(\{W\left(\delta_s\otimes\delta_x\right)\}_{s\in\mathcal S})\cap\conv(\{W\left(\delta_s\otimes\delta_{x'}\right)\}_{s\in\mathcal S})\neq\emptyset,
		\end{align}
		where $\conv$ denotes the convex hull, which is defined as
		\begin{align}\label{eqn:symm_max_error}
		\conv(W):=\left\{W:\lambda_s\geq 0\forall s\in\mathcal{S}\land\sum_{s\in\mathcal{S}}\lambda_s=1\right\},
		\end{align}
		where $W=\sum_{s\in\mathcal{S}}\lambda_sW\left(\delta_s\otimes\delta_x\right)$.
	\end{definition}
	\begin{definition}[Average Error Symmetrizability \cite{720535}]\label{def:sim_symm}
		An AVC $\mathcal{W}\in\mathcal{C}(\mathcal{X}\times\mathcal{S},\mathcal{Y})$ is called average error symmetrizable $(\overline{e}$-symm.$)$, if for some $U\in\mathcal{C}(\mathcal{X},\mathcal{S})$,
		\begin{align}\label{eqn:symm_avg_error}
		\sum\limits_{s\in\mathcal{S}}w(y\vert x,s)u(s\vert x')=	\sum\limits_{s\in\mathcal{S}}w(y\vert x',s)u(s\vert x),
		\end{align}
		for every $x,x'\in\mathcal{X}$, $y\in\mathcal{Y}$.
	\end{definition}
	The difference between ${e}^{\mathrm{max}}$-symm. and $\overline{e}$-symm. is visualized in Figure~\ref{fig:symmetrizability_max_avg} for $|\mathcal{X}|=4$.
	\setlength{\belowcaptionskip}{-1pt}
	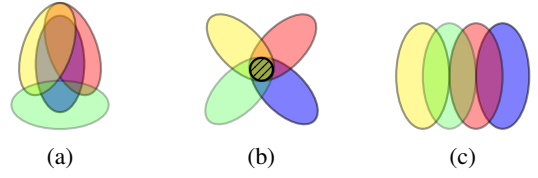
\begin{figure}
		\centering
		\begin{subfigure}[b]{0.16\textwidth}
			\centering
			\begin{tikzpicture}[thick,scale=0.32, every node/.style={transform shape}]
			\draw[fill=blue,opacity=.5] (0,0) ellipse (1cm and 2cm);
			\draw[fill=red,opacity=.4,rotate=20] (0.7,0.4) ellipse (1cm and 2cm);
			\draw[fill=yellow,opacity=.4,rotate=-20] (-0.7,0.4) ellipse (1cm and 2cm);
			\draw[fill=green,opacity=.25,rotate=90] (-1.7,0) ellipse (1cm and 2cm);
			\end{tikzpicture}
			\caption{}
			\label{fig:max}
		\end{subfigure}
		\begin{subfigure}[b]{0.13\textwidth}
			\centering
			\begin{tikzpicture}[thick,scale=0.28, every node/.style={transform shape}]
			\draw[fill=blue,opacity=.5,rotate=45] (-1,-2.5) ellipse (1cm and 2cm);
			\draw[fill=red,opacity=.4,rotate=-45] (1,0.4) ellipse (1cm and 2cm);
			\draw[fill=yellow,opacity=.4,rotate=45] (-1,0.4) ellipse (1cm and 2cm);
			\draw[fill=green,opacity=.25,rotate=-45] (1,-2.5) ellipse (1cm and 2cm);
			\draw[fill=black, pattern=north east lines] (0,-1.45) circle (0.54cm);
			\end{tikzpicture}
			\caption{}
			\label{fig:blume}
		\end{subfigure}
		\begin{subfigure}[b]{0.16\textwidth}
			\centering
			\begin{tikzpicture}[thick,scale=0.35, every node/.style={transform shape}]
			\draw[fill=blue,opacity=.5] (1.5,0) ellipse (1cm and 2cm);
			\draw[fill=red,opacity=.4] (0.5,0) ellipse (1cm and 2cm);
			\draw[fill=yellow,opacity=.4] (-1.5,0) ellipse (1cm and 2cm);
			\draw[fill=green,opacity=.25] (-0.5,0) ellipse (1cm and 2cm);
			\end{tikzpicture}
			\caption{}
			\label{fig:non_symm}
		\end{subfigure}
		\caption{Criteria for ${e}^{\mathrm{max}}$-symm. (\subref{fig:max}) and $\overline{e}$-symm. (\subref{fig:blume}). Depicted are the convex hulls induced by \eqref{eqn:symm_max_error} and \eqref{eqn:symm_avg_error}. Each region (red, green, blue, yellow) represents the set of states that can be created by the jammer given a particular input. Subfigure~(\subref{fig:max}) and (\subref{fig:blume}) show $\overline{e}$-symmetrizable AVCs, Subfigure~(\subref{fig:non_symm}) corresponds to a non-symmetrizable AVC.}
		\label{fig:symmetrizability_max_avg}
	\end{figure}
\end{section}
\begin{section}{Invertibility and Maximum Error Capacity}\label{sec:main_Results}
In this section we provide concepts that allow for reliable communication over symmetrizable AVCs. For these investigations, we need Ahlswede's ``Separation Lemma".
	\begin{lemma}[Separation Lemma \cite{Ahlswede1978}]\label{lem:separation_lemma}
	For the AVC $\mathcal{W}$ and the deterministic capacity under the maximum error criterion $C_\mathrm{d}^\mathrm{max}$ the following two statements are equivalent:
		\begin{enumerate}
		\item $C_\mathrm{d}^\mathrm{max}(\mathcal{W})>0$,
		\item $\mathcal{W}$ is non ${e}^{\mathrm{max}}$-symm. according to Definition~\ref{def:x-symm}.
		\end{enumerate}
	\end{lemma}
Moreover, we make use of the left-invertibility of a matrix which is defined in the following.
	\begin{definition}[Left-Invertibility]\label{def:left-invertibility}
		Let $A$ be an $m\times n$ matrix. We say that $A$ is left invertible if there exists a matrix $B$ of size $n\times m$ such that $BA=Id$.
	\end{definition}
Theorem~\ref{thm:inv_implies_symm} establishes a direct link between the symmetrizability of an AVC under the maximum error criterion and the invertibility of the corresponding channel matrix.
	\begin{theorem}\label{thm:inv_implies_symm}
	Let $\mathcal{S}$, $\mathcal{X}$ and $\mathcal{Y}$ be finite sets. Let $\mathcal{W}\in\mathcal{C}(\mathcal{X}\times\mathcal{S},\mathcal{Y})$ be left-invertible according to Definition~\ref{def:left-invertibility}. Then there exists $x,x'\in\mathcal{X}$, such that $\conv(\{W\left(\delta_x\otimes\delta_s\right)\}_{s\in\mathcal S})\cap\conv(\{W\left(\delta_{x'}\otimes\delta_{s}\right)\}_{s\in\mathcal S})=\emptyset$.	Thus $\mathcal{W}$ is non ${e}^{\mathrm{max}}$-symm. according to Definition~\ref{def:x-symm}.
	\end{theorem}
	\begin{proof}[Proof of Theorem~\ref{thm:inv_implies_symm}]
	Assume for contradiction that for all $x,x'\in\mathcal{X}$ the following holds:
		\begin{align}\label{eqn:first}
		\left\{\sum_{s\in\mathcal{S}}\lambda_sW\left(\delta_x\otimes\delta_s\right)\right\}\cap\left\{\sum_{s\in\mathcal{S}}\mu_sW\left(\delta_{x'}\otimes\delta_{s}\right)\right\}\neq\emptyset,
		\end{align}
	with $\lambda_s,\mu_s\geq 0$ with $\sum_{s\in\mathcal{S}}\lambda_s=1$ and $\sum_{s\in\mathcal{S}}\mu_s=1$. Equation~\eqref{eqn:first} implies that the following equation has to have a solution for every $x,x'\in\mathcal{X}$:
		\begin{align}\label{eqn:second}
		\sum_{s\in\mathcal{S}}\lambda_sW\left(\delta_x\otimes\delta_s\right)&=\sum_{s\in\mathcal{S}}\mu_sW\left(\delta_{x'}\otimes\delta_{s}\right).
		\end{align}
	Since $\mathcal{W}$ is left-invertible by assumption, \eqref{eqn:second} can be reformulated in the following way:
		\begin{align}\label{eqn:three}
		\sum_{s\in\mathcal{S}}\lambda_s\delta_x\otimes\delta_{s}&=\sum_{s\in\mathcal{S}}\mu_s\delta_{x'}\otimes\delta_{s}.
		\end{align}
	Taking the partial trace over the first subsystem of both sides of \eqref{eqn:three} leads to $\delta_{x}=\delta_{x'}$. Since $x\neq x'$ by assumption, the theorem is proven.
	\end{proof}
	\begin{remark}
	The converse is not necessarily true.
	\end{remark}
In this work we develop applications of the invertibility criterion allowing for a comprehensive understanding of the impact of receive diversity on communication over AVCs. Next, we prove a lower bound on the maximum error capacity of non-symmetrizable AVCs. This is a first step towards a more widespread use of maximum error criteria in the evaluation and development of future applications. Ahlswede and Wolfowitz determined the capacity for a non-symmetrizable AVC with binary output alphabet in \cite{ahlswede1970capacity}.
	\begin{lemma}[Maximum Error Capacity of the Binary Output AVC \cite{ahlswede1970capacity}]\label{lem:binary_AVC_cappa}
	Let $W\in\mathcal C(\mathcal X,\mathcal Y)$ be a non-symmetrizable AVC with input alphabet $\mathcal{X}$, binary output alphabet $\mathcal{Y}=\{0,1\}$ and state set $\mathcal{S}$. It holds $	C_\mathrm{d}^\mathrm{max}(\mathcal{W})=R_\mathrm{d}^\mathrm{max}(\mathcal{W})$, where
		\begin{align}\label{eqn:binary_output_AVC_cappa}
		R_\mathrm{d}^\mathrm{max}(\mathcal{W}):=\min_{T\in\mathcal{C}(\mathcal{X},\mathcal{S})}\max_{p\in\mathcal{P}^{(2)}(\mathcal{X})}I\left(p;W(Id\otimes T)\right).
		\end{align}
	\end{lemma}
In order to be able to provide a lower bound on the maximum error capacity, which we present later in this section, we need to ensure the existence of a quantization which maps the received output symbol to a smaller alphabet and preserves the property of non-symmetrizability of an AVC.
	\begin{lemma}[Non-Symmetrizability Preserving Quantization]\label{lem:bring_it_together}
	Let $\Delta^{(d)}$ denote the probability simplex of dimension $d\in\mathbb{N}$ defined as $\Delta^{(d)}:=\left\{x\in\mathbb{R}^{d+1}\middle|x\geq 0,\mathbf{1}^Tx=1\right\}$ with $d\in\mathbb{N}_{\geq 2}$ and $\mathbf{1}$ denotes the all ones vector $(1,\ldots,1)\in\mathbb{R}^d$. Let $\mathcal{A},\mathcal{B}\subset\Delta^{(d)}$ be two compact convex sets on the probability simplex for which $\mathcal{A}\cap\mathcal{B}=\emptyset$. Then there exist linear transformations $Q_i:\mathbb{R}^{d+1-i}\to\mathbb{R}^{d-i},Q_i(\Delta^{(d+1-i)})\subset\Delta^{(d-i)}$ for $i\in[1,d-2]$ such that for the sets after iteratively applying $Q_i$,
	\begin{align}
		\hat{\mathcal{A}}_{d-2}:=\{(Q_{d-2}\circ\ldots\circ Q_1)(a):a\in\mathcal{A}\},
	\end{align}
	and
	\begin{align}
		\hat{\mathcal{B}}_{d-2}:=\{(Q_{d-2}\circ\ldots\circ Q_1)(b):b\in\mathcal{B}\},
	\end{align}
	it still holds $\hat{\mathcal{A}}_{d-2}\cap\hat{\mathcal{B}}_{d-2}=\emptyset$.
	\end{lemma}
	The proof of Lemma~\ref{lem:bring_it_together} can be found in the Appendix. The following theorem presents a result for positivity conservation under binary quantization in non-symmetrizable AVCs and thus, via \eqref{eqn:binary_output_AVC_cappa}, delivers a lower bound on the maximum error capacity in situations where $C_\mathrm{d}^\mathrm{max}(\mathcal{W})$ is positive.
	\begin{theorem}(Lower Bound on the Maximum Error Capacity)\label{thm:when_positive_than_lower_bound}
	Let $\mathcal{X}$,$\mathcal{Y}$ and $\mathcal{S}$ be finite alphabets. Let $\mathcal{W}\in\mathcal C(\mathcal{X}\times\mathcal S,\mathcal Y)$ be an AVC. If $\mathcal{W}$ is non-symmetrizable, a lower bound on the deterministic maximum error capacity is given by
		\begin{align}\label{eq:lower_bound_general}
		C_\mathrm{d}^\mathrm{max}(\mathcal{W})\geq\max_{Q\in\mathcal{C}(\mathcal{Y},\{0,1\})}R_\mathrm{d}^\mathrm{max}(Q\circ\mathcal{W})>0.
		\end{align}
	\end{theorem}
	\begin{proof}[Proof of Theorem~\ref{thm:when_positive_than_lower_bound}]
	Since $\mathcal{W}$ is non-symmetrizable by assumption, Lemma~\ref{lem:separation_lemma} implies $C_\mathrm{d}^\mathrm{max}(\mathcal{W})>0$. According to Lemma~\ref{lem:bring_it_together}, there exists a $Q\in\mathcal{C}(\mathcal{Y},\{0,1\})$ such that
	\begin{align}
		\conv(\{W\left(\delta_x\otimes\delta_s\right)\}_{s\in\mathcal S})\cap\conv(\{W\left(\delta_{x'}\otimes\delta_{s}\right)\}_{s\in\mathcal S})=\emptyset,
	\end{align}
	implies $\conv(\{Q(\delta_y)\circ W\left(\delta_x\otimes\delta_s\right)\}_{s\in\mathcal S})\cap\conv(\{Q(\delta_y)\circ W\left(\delta_{x'}\otimes\delta_{s}\right)\}_{s\in\mathcal S})=\emptyset$. Then Lemma~\ref{lem:binary_AVC_cappa} provides a lower bound for $C_\mathrm{d}^\mathrm{max}(\mathcal{W})$.
	\end{proof}
\end{section}
\begin{section}{Block-Restricted Jamming}\label{sec:temporal_prop}
In practical situations there may exist several restrictions for the jammer in the AVC framework. These can, for example, be power limitations (cf.~\cite{2627}), delay-constraints (cf.~\cite{dey2010coding}) and/or local state constraints (cf.~\cite{arendt2017}). In order to clearly distinguish from previous works, we focus on a situation where the jammer is limited to large-scale adjustments concerning his state-selection abilities in time in this contribution. This constellation could, for example, be found in a block fading setting.
	\begin{definition}[Block-Restricted Jamming]\label{def:block-restr_jamming}
	Let $\mathcal W\in\mathcal C(\mathcal X\times\mathcal S,\mathcal Y)$ denote an AVC and let $J,n\in\nn$ with $n\geq J$. A jammer is said to be $J$-block-restricted if his choice of states is restricted to the set
	\begin{align}
		\mathcal S_{J,n}:=\left\{s^n\in\mathcal{S}:\left\lceil\frac{i}{J}\right\rceil=\left\lceil\frac{j}{J}\right\rceil\implies s_i=s_j\right\},
	\end{align}
	$\forall i,j\in\{1,\ldots,n\}$ with $i\neq j$. $\mathcal X_{J,n}$ is defined analogously. The maximum error probability of a code $\mathcal K_n$ with message set $\mathcal{M}_n$ and decoding sets $\mathcal{D}_m\subset\mathcal{Y}$ for which $\mathcal{D}_m\cap \mathcal{D}_{m'}=\emptyset$ for every $m\neq m'$ under $J$-block-restricted jamming is
		\begin{align}
		e^{\max}_{\mathrm{UA},J}(\mathcal K_n):=\max_{s^n\in\mathcal S_{J,n}}\max_{m\in\mathcal M_n}w^{\otimes n}\left(\mathcal D^{\complement}_m|x^n_m,s^n\right).
		\end{align}	
	Achievable rate and $C_{\mathrm{d},J}^\mathrm{\max}(\mathcal W)$ under $J$-block-restricted jamming are defined according to Definition~\ref{def:achievable_rate} and Definition~\ref{def:capacity}, respectively.	
	\end{definition}
	\begin{remark}
		Definition~\ref{def:block-restr_jamming} implicitly contains the assumption that transmitter and jammer are synchronized. The unsynchronized setting will be the subject of a separate study.
	\end{remark}
	The study of block-restricted jamming allows us to join the topics of invertibility, symmetrizability and capacity of an AVC in a practical meaningful way.
	\setlength{\belowcaptionskip}{-3pt}
	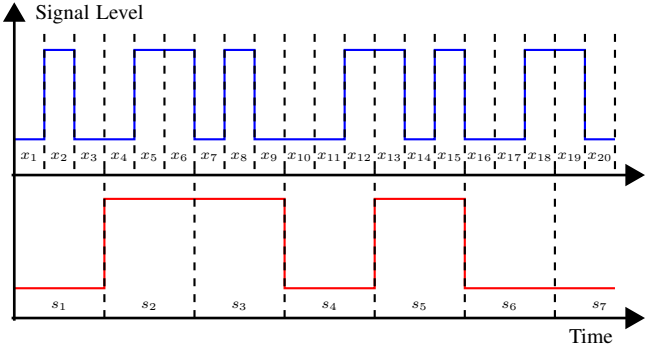
\begin{figure}
	\centering
		\begin{tikzpicture}[thick,scale=0.79, every node/.style={transform shape}]
		\draw[color=blue] (-3,-2) -- (-2.5,-2) -- (-2.5,-0.5) -- (-2,-0.5) -- (-2,-2) -- (-1.5,-2) -- (-1,-2) -- (-1,-0.5) -- (0,-0.5) -- (0,-2) -- (0.5,-2) -- (0.5,-0.5) -- (1,-0.5) -- (1,-2) -- (2.5,-2) -- (2.5,-0.5) -- (3.5,-0.5) -- (3.5,-2) -- (4,-2) -- (4,-0.5) -- (4.5,-0.5) -- (4.5,-2) -- (5,-2) -- (5.5,-2) -- (5.5,-0.5) -- (6.5,-0.5) -- (6.5,-2) -- (7,-2);
		\draw[color=red] (-3,-4.5) -- (-1.5,-4.5) -- (-1.5,-3) -- (0,-3) -- (1.5,-3) -- (1.5,-4.5) -- (3,-4.5) -- (3,-3) -- (4.5,-3) -- (4.5,-4.5) -- (6,-4.5) -- (7,-4.5);
		\draw[-triangle 60,line width=1pt, shorten <= -1pt] (-3,-5) -- (7.5,-5);
		\draw[-triangle 60,line width=1pt, shorten <= -1pt] (-3,-2.6) -- (7.5,-2.6);
		\draw[-triangle 60,line width=1pt, shorten <= -1pt] (-3,-5) -- (-3,0.3);
		\draw[color=black,style=dashed] (-1.5,-5) -- (-1.5,-2.5);
		\draw[color=black,style=dashed] (0,-5) -- (0,-2.5);
		\draw[color=black,style=dashed] (1.5,-5) -- (1.5,-2.5);
		\draw[color=black,style=dashed] (3,-5) -- (3,-2.5);
		\draw[color=black,style=dashed] (4.5,-5) -- (4.5,-2.5);
		\draw[color=black,style=dashed] (6,-5) -- (6,-2.5);
		\draw[color=black,style=dashed] (-2.5,-2.5) -- (-2.5,-0.1);
		\draw[color=black,style=dashed] (-2,-2.5) -- (-2,-0.1);
		\draw[color=black,style=dashed] (-1.5,-2.5) -- (-1.5,-0.1);
		\draw[color=black,style=dashed] (-1,-2.5) -- (-1,-0.1);
		\draw[color=black,style=dashed] (-0.5,-2.5) -- (-0.5,-0.1);
		\draw[color=black,style=dashed] (0,-2.5) -- (0,-0.1);
		\draw[color=black,style=dashed] (0.5,-2.5) -- (0.5,-0.1);
		\draw[color=black,style=dashed] (1,-2.5) -- (1,-0.1);
		\draw[color=black,style=dashed] (1.5,-2.5) -- (1.5,-0.1);
		\draw[color=black,style=dashed] (2,-2.5) -- (2,-0.1);
		\draw[color=black,style=dashed] (2.5,-2.5) -- (2.5,-0.1);
		\draw[color=black,style=dashed] (3,-2.5) -- (3,-0.1);
		\draw[color=black,style=dashed] (3.5,-2.5) -- (3.5,-0.1);
		\draw[color=black,style=dashed] (4,-2.5) -- (4,-0.1);
		\draw[color=black,style=dashed] (4.5,-2.5) -- (4.5,-0.1);
		\draw[color=black,style=dashed] (5,-2.5) -- (5,-0.1);
		\draw[color=black,style=dashed] (5.5,-2.5) -- (5.5,-0.1);
		\draw[color=black,style=dashed] (6,-2.5) -- (6,-0.1);
		\draw[color=black,style=dashed] (6.5,-2.5) -- (6.5,-0.1);
		\draw[color=black,style=dashed] (7,-2.5) -- (7,-0.1);
		\node (A1) at (-2.75,-2.3)[text width=0.5cm,align=center] {\textcolor{black}{{\scriptsize $x_1$}}};
		\node (A2) at (-2.25,-2.3)[text width=1cm,align=center] {\textcolor{black}{{\scriptsize $x_2$}}};
		\node (A3) at (-1.75,-2.3)[text width=2cm,align=center] {\textcolor{black}{{\scriptsize $x_3$}}};
		\node (A4) at (-1.25,-2.3)[text width=2cm,align=center] {\textcolor{black}{{\scriptsize $x_4$}}};
		\node (A5) at (-0.75,-2.3)[text width=2cm,align=center] {\textcolor{black}{{\scriptsize $x_5$}}};
		\node (A6) at (-0.25,-2.3)[text width=2cm,align=center] {\textcolor{black}{{\scriptsize $x_6$}}};
		\node (A7) at (0.25,-2.3)[text width=2cm,align=center] {\textcolor{black}{{\scriptsize $x_7$}}};
		\node (A8) at (0.75,-2.3)[text width=2cm,align=center] {\textcolor{black}{{\scriptsize $x_8$}}};
		\node (A9) at (1.25,-2.3)[text width=2cm,align=center] {\textcolor{black}{{\scriptsize $x_9$}}};
		\node (A10) at (1.75,-2.3)[text width=2cm,align=center] {\textcolor{black}{{\scriptsize $x_{10}$}}};
		\node (A11) at (2.25,-2.3)[text width=2cm,align=center] {\textcolor{black}{{\scriptsize $x_{11}$}}};
		\node (A12) at (2.75,-2.3)[text width=2cm,align=center] {\textcolor{black}{{\scriptsize $x_{12}$}}};
		\node (A13) at (3.25,-2.3)[text width=2cm,align=center] {\textcolor{black}{{\scriptsize $x_{13}$}}};
		\node (A14) at (3.75,-2.3)[text width=2cm,align=center] {\textcolor{black}{{\scriptsize $x_{14}$}}};
		\node (A15) at (4.25,-2.3)[text width=2cm,align=center] {\textcolor{black}{{\scriptsize $x_{15}$}}};
		\node (A16) at (4.75,-2.3)[text width=2cm,align=center] {\textcolor{black}{{\scriptsize 	$x_{16}$}}};
		\node (A17) at (5.25,-2.3)[text width=2cm,align=center] {\textcolor{black}{{\scriptsize $x_{17}$}}};
		\node (A18) at (5.75,-2.3)[text width=2cm,align=center] {\textcolor{black}{{\scriptsize $x_{18}$}}};
		\node (A19) at (6.25,-2.3)[text width=2cm,align=center] {\textcolor{black}{{\scriptsize $x_{19}$}}};
		\node (A20) at (6.75,-2.3)[text width=2cm,align=center] {\textcolor{black}{{\scriptsize $x_{20}$}}};
		\node (B1) at (-2.25,-4.8)[text width=1cm,align=center] {\textcolor{black}{{\scriptsize $s_1$}}};
		\node (B2) at (-0.75,-4.8)[text width=2cm,align=center] {\textcolor{black}{{\scriptsize $s_2$}}};
		\node (B3) at (0.75,-4.8)[text width=2cm,align=center] {\textcolor{black}{{\scriptsize $s_3$}}};
		\node (B4) at (2.25,-4.8)[text width=2cm,align=center] {\textcolor{black}{{\scriptsize $s_4$}}};
		\node (B5) at (3.75,-4.8)[text width=2cm,align=center] {\textcolor{black}{{\scriptsize $s_5$}}};
		\node (B6) at (5.25,-4.8)[text width=2cm,align=center] {\textcolor{black}{{\scriptsize $s_6$}}};
		\node (B7) at (6.75,-4.8)[text width=2cm,align=center] {\textcolor{black}{{\scriptsize $s_7$}}};
		\node (B7) at (6.6,-5.3)[text width=2cm,align=center] {\textcolor{black}{{Time}}};
		\node (B7) at (-1.75,0.1)[text width=2cm,align=center] {\textcolor{black}{{Signal Level}}};
		\end{tikzpicture}
	\caption{Block-restricted jamming for $J=3$ and $|\mathcal{X}|=|\mathcal{S}|=2$.}
	\label{fig:block_constrained_jamming}
	\end{figure}
The block-restricted jamming scenario is displayed in Figure~\ref{fig:block_constrained_jamming}. For reliable communication based on application of Theorem~\ref{thm:inv_implies_symm} we require $|\mathcal{Y}|\geq|\mathcal{X}|$. Usually, communication systems are designed such that $|\mathcal{Y}|=|\mathcal{X}|$. Nevertheless, in the AVC setting, the input alphabet is ``fanned" out by the influence of the jammer. One question is whether it is possible to recover from that unsatisfactory constellation. A partial answer is given by \cite[Theorem~1]{1412.5831} showing that sufficiently enlarging the output alphabet ensures invertibility.
	\begin{lemma}\label{lem:inversification}
	Let $W\in\mathcal C(\mathcal X,\mathcal Y)$ satisfy $W(\delta_x)\neq W(\delta_{x'})$ for all $x\neq x'\in\cX$. Then $J\geq |\mathcal X|\cdot(|\mathcal Y|-1)$ is sufficient for $W^{(J)}$ defined via $W^{(J)}(\delta_x):=\otimes_{i=1}^JW(\delta_x)$ to be invertible as a map from $\cP^{(J)}(\mathcal X)$ to $\cP\left(\mathcal Y^J\right)$.
	\end{lemma}
The following theorem makes use of Lemma~\ref{lem:inversification} to obtain a lower bound on the maximum error capacity under block-restricted jamming. In the remainder, $\mathcal{X}$, $\mathcal{Y}$ and $\mathcal{S}$ denote finite sets and $W\in\mathcal C(\mathcal X\times\mathcal S,\mathcal Y)$ an AVC.
\begin{theorem}\label{thm:inversification_AVC}
	Let the AVC $W$ have the property $w(\cdot|x,s)\neq w\left(\cdot\middle|x',s'\right)$ for all $x,x'\in\mathcal X$ and $s,s'\in\mathcal S$ satisfying $(x,s)\neq(x',s')$. Let $E_\mathcal{X}\in\mathcal{C}(\mathcal{X},\mathcal{X}_{J})$ and $E_\mathcal{S}\in\mathcal{C}(\mathcal{S},\mathcal{S}_{J})$ with $e_x\left(x^{J}\middle|x\right):=\prod_{j=1}^{J}\delta(x_{j},x)$ and $e_s\left(s^{J}\middle|s\right):=\prod_{j=1}^{J}\delta(s_{j},s)$, respectively. There exists $J\in\nn$ such that $C_{\mathrm d,J}^{\max}(\mathcal W)>0$. If $J\geq|\mathcal X|\cdot|\mathcal S|\cdot(|\mathcal Y|-1)$, then
	\begin{align}
	C_{\mathrm d,J}^{\max}(\mathcal W)\geq \frac{1}{J}R_{\mathrm d}^{\max}\left(Q^\star\circ W^{\otimes J}\circ (E_\mathcal{X}\otimes E_\mathcal{S})\right),
	\end{align}
	where $Q^\star\in\mathcal{C}\left(\mathcal{Y}^J,\{0,1\}\right)$ is the the optimal quantizer with respect to $W$ solving the outer maximization in \eqref{eq:lower_bound_general}.
\end{theorem}
\begin{remark}
	Observe that the restriction on the channels in Theorem~\ref{thm:inversification_AVC} only excludes very specific constellations of channel parameters and thus a small set of AVCs. Moreover, notice that we exploit the repetitive usage of the communication channel during a constant jammer signal which is a fundamental difference compared to simply enlarging the input alphabet of the transmitter which does not necessarily guarantee non-symmetrizability of an AVC.
\end{remark}
\begin{remark}
	The wide spread application of Shannon and coding theory to modern communication systems has rightfully lead to abandoning of repetition coding. Theorem~\ref{thm:inversification_AVC} explains why smaller numbers of repetitions could be an enabler for reliable communication in ad-hoc networks.
\end{remark}
\begin{proof}[Proof of Theorem~\ref{thm:inversification_AVC}]
	Lemma~\ref{lem:inversification} ensures invertibility of the matrix associated to $\mathcal{W}$. Theorem~\ref{thm:inv_implies_symm} then implies non $e^\mathrm{max}$-symm. The result follows directly from Theorem~\ref{thm:when_positive_than_lower_bound}.
\end{proof}
\begin{theorem}\label{thm:inversification_AVC_upper}
	Let the preliminaries be as in Theorem~\ref{thm:inversification_AVC}. An upper bound for $C_{\mathrm d,J}^{\max}(\mathcal W)$ is given by
		\begin{align}
		C_{\mathrm d,J}^{\max}(\mathcal W)\leq\min_{q\in\mathcal{P}(\mathcal S)}\max_{p\in\mathcal{P}(\mathcal{X}^{J})}\frac{1}{J}I\left(p\otimes q_J;W^{\otimes J}(Id\otimes E_\mathcal{S})\right).
		\end{align}
	\end{theorem}
	\begin{proof}[Proof of Theorem~\ref{thm:inversification_AVC_upper}]
	Assume that there exists a ($n$,$\mathcal{M}_n$)-code for the AVC under $J$-block-restricted jamming such that
	\begin{align}
	\frac{1}{n}\log|\mathcal{M}_n|>R-\eps_n,
	\end{align}
	and
	\begin{align}\label{eqn:error_criterion}
	e^{\max}_{\mathrm{UA},J}(\mathcal K_n)<\eps_n.
	\end{align}
	It follows that there exists a code with the property that $\overline{e}_\mathrm{UA}(\mathcal{K}_n,q)<\eps$, where $\overline{e}_\mathrm{UA}(\mathcal{K}_n,q)$ defines the average error of the DMC induced by the i.i.d. jamming strategy $q_J\in\mathcal{P}(\mathcal{S_J})$ under the usage of the code $\mathcal{K}_n$, that is,
	\begin{align}
	\begin{split}
	\overline{e}_\mathrm{UA}(\mathcal{K}_n,q_J):=&\frac{1}{|\mathcal{M}_n|}\sum_{m\in\mathcal{M}_n}\left(w_{q_J}^{\otimes \lfloor n/J\rfloor}\otimes w_{q_J}^{\otimes r}\left(\mathcal{D}^{\complement}_m|x_m^n\right)\right),
	\end{split}
	\end{align}
	where $r$ is the remainder of the division $n/J$. Thus
	\begin{align}
	nR&\leq H(M)\\
	&=H\left(M\middle|\hat{M}\right)+I\left(M;\hat{M}\right)\\
	&\leq 1+nR\overline{e}_{\mathrm{UA},J}+I\left(M;\hat{M}\right)\label{eqn:fano_ineq}\\
	&\leq 1+nR\eps_n+I\left(X^n;Y^n\right)\label{eqn:data_proc_ineq},
	\end{align}
		where \eqref{eqn:fano_ineq} follows from Fano's Inequality and \eqref{eqn:data_proc_ineq} from Data Processing Inequality and \eqref{eqn:error_criterion}. In the following, we concentrate on $I\left(X^n\middle|Y^n\right)$. Let
		\begin{align}
		\begin{split}
			\hat{X}^{\lceil n/J\rceil}:=&X_1,\ldots,X_{J},X_{J+1},\ldots,X_{2J},X_{2J+1},\\
			&\ldots,X_{\lfloor n/J\rfloor},X_{\lfloor n/J\rfloor+1},\ldots,X_n,
		\end{split}
		\end{align}
		 and
		 \begin{align}
		 \begin{split}
		 	\hat{Y}^{\lceil n/J\rceil}:=&Y_1,\ldots,Y_{J},Y_{J+1},\ldots,Y_{2J},Y_{2J+1},\\
		 	&\ldots,Y_{\lfloor n/J\rfloor},Y_{\lfloor n/J\rfloor+1},\ldots,Y_n.
		 \end{split}
		 \end{align}
		 Let $Y^{i-1}:=Y_1,Y_2,\ldots,Y_{i-1}$ with $i\in\nn$. For the DMC generated from the block-restricted AVC for a fixed i.i.d strategy $q_J\in\mathcal P(\mathcal S_J)$ of the Jammer, it holds
		 \begin{align}
		 I\left(X^n;Y^n\right)&=I\left(\hat{X}^{\lceil n/J\rceil};\hat{Y}^{\lceil n/J\rceil}\right)\\
		 &=H\left(\hat{Y}^{\lceil n/J\rceil}\right)-H\left(\hat{Y}^{\lceil n/J\rceil}\middle|\hat{X}^{\lceil n/J\rceil}\right)\\
		 &=H\left(\hat{Y}^{\lceil n/J\rceil}\right)-\sum_{k=1}^{\lfloor n/J\rfloor+1}H\left(\hat{Y}_i\middle|Y^{i-1},\hat{X}^n\right)\label{eqn:chain_rule_eq}\\
		 &\leq\sum_{k=1}^{\lfloor n/J\rfloor+1}H\left(\hat{Y}_k\right)-\sum_{k=1}^{\lfloor n/J\rfloor+1}H\left(\hat{Y}_k\middle|\hat{X}_k\right)\label{eqn:block_restr_eq}\\
		 &=\sum_{k=1}^{\lfloor n/J\rfloor+1}I\left(\hat{Y}_k;\hat{X}_k\right)\\
		 \begin{split}
		 &\leq\left\lfloor \frac{n}{J}\right\rfloor\max_{p\in\mathcal P\left(\mathcal X^J\right)}I\left(p\otimes q_J,W^{\otimes J}\right)\\&\quad+\log_2\left(X^J\right)\label{eqn:DMC_cappa_eq},
		 \end{split}
		 \end{align}
		 where \eqref{eqn:chain_rule_eq} follows by chain rule, \eqref{eqn:block_restr_eq} because the channel is memoryless from block to block ($\hat{Y}_i\text{---}\hat{X}_i\text{---}\hat{Y}_1,\ldots\hat{Y}_{i-1}$), that is, it holds $P_{\hat{Y}_n|\hat{Y}_1,\ldots \hat{Y}_n}=P_{\hat{Y}_n|\hat{Y}_{n-1}}$. Since \eqref{eqn:DMC_cappa_eq} holds for every choice $q_J\in\mathcal P(\mathcal S_J)$, we get
		 \begin{align}
		 \begin{split}
		 nR\leq 1&+nR\eps_n+\left\lfloor \frac{n}{J}\right\rfloor\max_{p\in\mathcal P(\mathcal X^J)}I\left(p\otimes q_J,W^{\otimes J}\right)\\&+\log_2\left(X^J\right).
		 \end{split}
		 \end{align}
		 We observe that (14) holds for all $q_J\in\mathcal P(\mathcal S_J)$. Thus, upon dividing by $n$ and minimizing over $q_J$ we obtain
		 \begin{align}
		 \begin{split}
		 R\leq \frac{1}{n}&+R\eps_n+\frac{\log_2\left(X^J\right)}{n}\\&+\left\lfloor \frac{n}{J}\right\rfloor\max_{p\in\mathcal P(\mathcal X^J)}\min_{q_J\in\mathcal{P}(\mathcal{S_J})}I\left(p\otimes q_J,W^{\otimes J}\right)\label{eq:rate_for_block_DMC}.
		 \end{split}
		 \end{align}
		 Now, we let $n\longrightarrow\infty$ in \eqref{eq:rate_for_block_DMC} and use the fact that the mutual information is convex in the channel parameters and concave in the input distribution resulting in
		 \begin{align}
		 R&\leq\min_{q_J\in\mathcal{P}(\mathcal{S_J})}\max_{p\in\mathcal{P}(\mathcal{X^J})}\frac{1}{J}I\left(p\otimes q_J,W^{\otimes J}\right)\\
		&=\min_{q\in\mathcal{P}(\mathcal{S})}\max_{p\in\mathcal{P}(\mathcal{X^J})}\frac{1}{J}I\left(p\otimes q,W^{\otimes J}\circ(Id\otimes E_{\mathcal S})\right),
		 \end{align}
		 which completes the converse proof.
	\end{proof}
\end{section}

\begin{section}{Conclusion}\label{sec:discussion_and_conclusion}
	In this contribution we developed an explicit connection relating the invertibility of an AVC channel matrix to the symmetrizability under the maximum error criterion. With the help of that relation we provided a lower bound on the maximum error capacity for invertible AVCs. We showed that reliable communication over symmetrizable AVCs under block-restricted jamming is possible exploiting time diversity by using repetition coding. This result guarantees reliable information transmission under all possible prevailing circumstances. Further we provided a lower and an upper bound on the maximum error capacity of the block-restricted AVC.
\end{section}
\begin{acknowledgement}
	C.A. thanks P. Fertl from the BMW Group. Funding is acknowledged from the DFG via grant BO 1734/20-1, the BMBF via grants 01BQ1050 and 16KIS0118 (H.B.), the BMWi and ESF via grant 03EFHSN102 (J.N.).
\end{acknowledgement}
\appendix
\addcontentsline{toc}{section}{Appendices}
		The Appendix mainly concentrates on the proof of Lemma~\ref{lem:bring_it_together}. First, Lemma~\ref{lem:shrinking and rotation} ensures that the quantized output distribution remains in $\mathcal{P}(\mathcal{Y})$. Second, Lemma~\ref{lem:projection} guarantees the existence of a disjointness preserving projection from $\Delta^{(d)}$ to $\Delta^{(d-1)}$ ensuring conservation of non $e^\mathrm{max}$-symm.
		\begin{lemma}[Shrinking and Rotation]\label{lem:shrinking and rotation}
		 Let $\Delta^{(d)}$ denote the probability simplex of dimension $d\in\mathbb{N}$ defined in Lemma~\ref{lem:binary_AVC_cappa}. For $x\in\Delta^{(d)}$, let $t(x):=\left|\left|x-\frac{1}{d}\sum_{i=0}^{d-1}\delta_i\right|\right|_2$ denote the distance of a point on the probability simplex to its center. Let $O$ be an arbitrary rotation for which it holds that $O(\mathbf{1})=\mathbf{1}$. Let {$N_{\lambda(d)}=\lambda(d)\cdot Id+(1-\lambda(d))T$} where $T:=\frac{1}{d}\sum_{i,j=0}^{d-1}E_{i,j}$. Then there exists $\lambda(d)\in[0,1]$ such that for every $x\in\Delta^{(d)}$ and every rotation $O$ with the previously mentioned properties, $N_{\lambda(d)}O(x)\in\Delta^{(d)}$.
		\end{lemma}
		\begin{proof}[Proof of Lemma~\ref{lem:shrinking and rotation}]
		First, consider the maximum distance $t_\mathrm{max}(d)$ of a point on the probability simplex of dimension $d\in\mathbb{N}_{\geq 2}$ to its origin in the $||\cdot||_2$ which is
		\begin{align}
		t_\mathrm{max}(d)&=\sqrt{\left(1-d^{-1}\right)^2+(d-1)d^{-2}}\\
		&=\sqrt{d-1}/\sqrt{d}.
		\end{align}
		The minimum distance $t_\mathrm{min}(d)$ (points on the boundary of $\Delta^{(d)}$ with exactly one coordinate being zero) reads as
		\begin{align}
		t_\mathrm{min}(d)&=\sqrt{d^{-2}+(d-1)\left(d^{-1}-(d-1)^{-1}\right)^2}\\
		&=1/\sqrt{d(d-1)}.
		\end{align}
		Now choose $\lambda(d)$ such that $\lambda(d)t_\mathrm{max}(d)\geq t_\mathrm{min}(d)$. Without loss of generality, let $\lambda(d)$ be chosen such that equality holds in the previous inequality, that is, $\lambda(d)=t_\mathrm{min}(d)/t_\mathrm{max}(d)$.
		Computing $N_{\lambda(d)}(x)$ gives
		\begin{align}
		N_{\lambda(d)}(x)=\lambda(d)\sum_{i=0}^{d-1}x_i\delta_i+(1-\lambda(d))\frac{1}{d}\sum_{i,j=0}^{d-1}x_j\delta_i.
		\end{align}
		Since $\sum_{i,j=0}^{d-1}x_j\delta_i=\sum_{i=0}^{d-1}\delta_i$, because $x_1+\dots+x_d=1$, it follows
		\begin{align}
		N_{\lambda(d)}(x)&=\left(\lambda(d)\left(x_i-\frac{1}{d}\right)+\frac{1}{d}\right)\sum_{i=0}^{d-1}\delta_i\\
		&=\lambda(d)(x-\pi)+\pi\label{eq:N_d_and_pi},
		\end{align}
		with $\pi=\frac{1}{d}(\delta_0+\delta_1+\dots+\delta_{d-1})$. Let $x$ satisfy $\|x-\pi\|\leq t_\mathrm{max}(d)$. Note that such vectors especially arise from application of an arbitrary rotation $O$ with the property $O(\pi)=\pi$ to an element of $\Delta^{(d)}$. Let $p=O(x)$ for some $x\in\Delta^{(d)}$ and an orthogonal transformation satisfying $O(\pi)=(\pi)$. Then for the distance between $N_{\lambda(d)}(p)$ and the center $\pi$ of $\Delta^{(d)}$ using \eqref{eq:N_d_and_pi} it holds
		\begin{align}
		\left|\left|N_{\lambda(d)}(p)-\pi\right|\right|_2
		&=\lambda(d)||p-\pi||_2\\
		&=\lambda(d)||O(x-\pi)||_2\\
		&\leq\lambda(d) t_\mathrm{max}(d)\label{eq:invariant_under_rotation}\\
		&=t_\mathrm{min}(d),
		\end{align}
		where \eqref{eq:invariant_under_rotation} follows, inter alia, from the property of the $||\cdot||_2$ being invariant under rotation. This proves the lemma.
	\end{proof}		
		\begin{lemma}[Disjointness Preserving Projection]\label{lem:projection}
		Let $\Delta^{(d)}$ denote the probability simplex of dimension $d\in\mathbb{N}$ defined in Lemma~\ref{lem:binary_AVC_cappa}. Let, for $v\neq 0$ and $c\in\mathbb{R}$, $\mathcal{H}:=\left\{x\in\mathbb{R}^d\middle|v^Tx=c\right\}$ be a hyperplane oriented such that $\delta_{d-1},\sum_{i=0}^{d-2}p_i\delta_i\in\mathcal{H}$ with fixed $p_i>0$ with $\sum_{i=0}^{d-2}p_i=1$. There exists a projection $V$ and a hyperplane $\tilde{\mathcal{H}}:=\left\{x\in\mathbb{R}^{d-1}\middle|\tilde{v}^Tx=c\right\}$ such that for every $q\in\Delta^{(d)}$ with $q=\mu\delta_{d-1}+(1-\mu)\left(\sum_{i=0}^{d-2}p'_i\delta_i\right)$, $p_i'\neq p_i$ and $p_i'\in[0,1]$ for which $\sum_{i=0}^{d-2}p'_i=1$, $\tilde{v}^TV(q)<c$ if $v^Tq<c$ and $\tilde{v}^TV(q)>c$ if $v^Tq>c$.
		\end{lemma}
		\begin{proof}[Proof of Lemma~\ref{lem:projection}]
		Define the projection $V$ in the following way: $V\left(\delta_{d-1}\right):=\sum_{i=1}^{d-2}p_{i}\delta_i=:p$. Applying the projection $V$ to $q$ leads to
		\begin{align}
		V(q)&=\mu\left(\sum_{i=0}^{d-2}p_i\delta_i\right)+(1-\mu)\left(\sum_{i=0}^{d-2}p'_i\delta_i\right)\\
		&=\sum_{i=0}^{d-2}\delta_i\left(\mu p_i+(1-\mu)p'_i\right).
		\end{align}
		By definition, $\tilde{v}=V(v)=\sum_{i=0}^{d-2}v_i\delta_i$. Thus, it holds
		\begin{align}
		\tilde{v}^TV(q)&=\sum_{i=0}^{d-2}\left(\mu v_ip_i+(1-\mu)v_ip'_i\right)\\
		&=\mu\left(\sum_{i=0}^{d-2}v_ip_i\right)+(1-\mu)\left(\sum_{i=0}^{d-2}v_ip'_i\right)\\
		&=\mu c+(1-\mu)\left(\sum_{i=0}^{d-2}v_ip'_i\right).\label{eqn:last_step_before_absch}
		\end{align}
		Without loss of generality, let $v^Tq<c$ such that $c>v^Tq=\mu c+(1-\mu)\sum_{i=0}^{d-2}v_ip_i'$ implies $(1-\mu)\sum_{i=0}^{d-2}v_ip'_i<(1-\mu)c$. Then by \eqref{eqn:last_step_before_absch} $\tilde{v}^TV(q)<\mu c+(1-\mu)c=c$ and for $v^Tq>c$ it holds $\tilde{v}^TV(q)>c$.
	\end{proof}
	\begin{proof}[Proof of Lemma~\ref{lem:bring_it_together}]
	Recall that $\mathcal{A}\cap\mathcal{B}=\emptyset$. First, set $\mathcal{A}_1:=\mathcal{A}$ and $\mathcal{B}_1:=\mathcal{B}$. Subsequently apply the following procedure: For the $i$-th step with $i\in[1,d-2]$, choose a shrinking operation $N_{i,\lambda_i(d-i+1)}$ as specified in Lemma~\ref{lem:shrinking and rotation}. By Lemma~\ref{lem:shrinking and rotation} we know that there exists a $\lambda_i(d-i+1)\in[0,1]$ such that for every $x\in\Delta^{(d+1-i)}$, $P_i(x)\in\Delta^{(d+1-i)}$ with $P_i=O\circ N_{i,\lambda_i(d)}$ for any rotation $O$ for which it holds $O(\mathbf{1})=\mathbf{1}$. Define the shrinken sets $\tilde{\mathcal{A}_i}$ and $\tilde{\mathcal{B}_i}$ as follows: $\tilde{\mathcal{A}_i}:=\left\{N_{i,\lambda_i(d)}(a_i):a_i\in\mathcal{A}_i\right\}$, $\tilde{\mathcal{B}_i}:=\left\{N_{i,\lambda_i(d)}(b_i):b_i\in\mathcal{B}_i\right\}$. Since shrinking is an invertible operation, it still holds $\tilde{\mathcal{A}_i}\cap\tilde{\mathcal{B}_i}=\emptyset$. Thus by the Separating Hyperplane Theorem \cite[Section 2.5.1]{boyd2004convex} there exists a hyperplane $\mathcal{H}_i$ with normal $v_i\neq 0$ and $c_i$ dividing $\tilde{\mathcal{A}_i}$ and $\tilde{\mathcal{B}_i}$. Without loss of generality, assume that it holds $v_i^T\tilde{a_i}<c_i$ for all $\tilde{a_i}\in\tilde{\mathcal{A}_i}$ and $v_i^T\tilde{b_i}>c_i$ for all $\tilde{b_i}\in\tilde{\mathcal{B}_i}$. Now choose $O$, according to Lemma~~\ref{lem:shrinking and rotation}, to be a particular rotation $O_i$ such that for $\tilde{\mathcal H_i}:=\left\{O_i(x):x\in\mathcal{H}_i\right\}=\left\{x\in\mathbb{R}^{d+1-i}:\tilde{v}_i^Tx=c_i\right\}\cap\Delta^{(d+1-i)}$,
	for the $d-i$-th entry of $\tilde{v}_i$ it holds $\tilde{v}_{i,d-i}=c_i$. From this it follows $\delta_{d-i}\in\tilde{\mathcal H_i}\land \sum_{j=0}^{d-1-i}p_{i,j}\delta_j\in\tilde{\mathcal H_i}$ for fixed $p_{i,j}$ with $\sum_{j=0}^{d-2}p_{i,j}=1$. Next define a projection $V_i(\delta_{d-i}):=\sum_{j=1}^{d-1-i}p_{i,j}\delta_j=:p_i$ with $p_{i,1},\ldots, p_{i,d-1-i}>0$ and $\tilde{v}_i^Tp_i=c_i$ and set $Q_i=P_i\circ V_i$. Since $Q_i$ is a concatenation of affine transformations, convexity of both sets is preserved under $Q_i$. Then Lemma~\ref{lem:projection} guarantees the existence of $\hat{\mathcal{H}_i}:=\left\{x\in\mathbb{R}^{d-i}:\hat{v}_i^Tx=c_i\right\}\cap\Delta^{(d-i)}$ separating the sets $\hat{\mathcal{A}_i}:=\left\{Q_i(a_i):a_i\in\mathcal{A}_i\right\}$ and $\hat{\mathcal{B}_i}:=\left\{Q_i(b_i):b_i\in\mathcal{B}_i\right\}$ after transformation. Now set $A_{i+1}:=\hat{\mathcal{A}_i}$ and $B_{i+1}:=\hat{\mathcal{B}_i}$. Iteratively apply the previous steps increasing $i$ by one until $i=d-1$ such that it then holds $\hat{\mathcal{A}}_{d-2},\hat{\mathcal{B}}_{d-2}\subset\Delta^{(2)}$.
	\end{proof}
\bibliographystyle{IEEEtran}
\bibliography{bibmax}
\end{document}